\documentclass[english]{lipics-v2021}
\nolinenumbers

\usepackage[utf8]{inputenc}

\usepackage{microtype,hyperref,graphicx}
   \hypersetup{%
      breaklinks,%
      colorlinks=true,%
      urlcolor=[rgb]{0.25,0.0,0.0},%
      linkcolor=[rgb]{0.5,0.0,0.0},%
      citecolor=[rgb]{0,0.2,0.445},%
      filecolor=[rgb]{0,0,0.4},
      anchorcolor=[rgb]={0.0,0.1,0.2}%
   }
   
\usepackage{makecell}

\def\Patrascu{P\u atra\c scu}
\def\Matousek{Matou\v{s}ek}
\def\Holder{H\"{o}lder}
\newcommand{\R}{\mathbb{R}}

\newcommand{\tildeO}{\Tilde{O}}

\newcommand{\calD}{\mathcal{D}}

\newcommand{\calR}{\mathcal{R}}

\newcommand{\eps}{\varepsilon}
\newcommand{\OO}{\Tilde{O}}
\newcommand{\OOO}{O^*}
\newcommand{\IGNORE}[1]{}

\newtheorem{subproblem}[theorem]{Subproblem}

\title{Dynamic Geometric Connectivity in the Plane with Constant Query Time}
   
\author{Timothy M. Chan}{Department of Computer Science, University of Illinois at Urbana-Champaign, USA}{tmc@illinois.edu}{https://orcid.org/0000-0002-8093-0675}{Work supported by NSF Grant CCF-2224271.}
\author{Zhengcheng Huang}{Department of Computer Science, University of Illinois at Urbana-Champaign, USA}{zh3@illinois.edu}{}{}

\titlerunning{Dynamic Geometric Connectivity in the Plane with Constant Query Time}
\authorrunning{T.\,M. Chan and Z. Huang}

\Copyright{Timothy M. Chan and Zhengcheng Huang}



\begin{document}

\maketitle

\begin{abstract}
We present the first fully dynamic connectivity data structures for geometric intersection graphs achieving \emph{constant} query time and \emph{sublinear} amortized update time for many classes of geometric objects in 2D\@.  Our data structures can answer connectivity queries between two objects, as well as ``global'' connectivity queries (e.g., deciding whether the entire graph is connected). Previously, the data structure by Afshani and Chan (ESA'06) achieved such bounds only in the special case of axis-aligned line segments or rectangles but did not work for arbitrary line segments or disks, whereas the data structures by Chan, P\u atra\c scu, and Roditty (FOCS'08) worked for more general classes of geometric objects but required $n^{\Omega(1)}$ query time and could not handle global connectivity queries.

Specifically, we obtain new data structures with $O(1)$ query time and amortized update time near $n^{4/5}$, $n^{7/8}$, and $n^{20/21}$ for axis-aligned line segments, disks, and arbitrary line segments respectively. Besides greatly reducing the query time, our data structures also improve the previous update times for axis-aligned line segments by Afshani and Chan (from near $n^{10/11}$ to $n^{4/5}$) and for disks by Chan, P\u atra\c scu, and Roditty (from near $n^{20/21}$ to $n^{7/8}$).
\end{abstract}

\section{Introduction}

\emph{Dynamic graph connectivity}---maintaining an undirected graph under edge insertions and deletions to answer connectivity queries---is a popular topic in data structure design and dynamic graph algorithms~\cite{HenzingerK99,HolmLT01,PatrascuT07,ChuzhoyGLNPS20,HuangHKPT23}. At STOC'02, Chan~\cite{Chan06} initiated the study of dynamic connectivity in \emph{geometric} settings: 
\begin{quote}
    Maintain a set of $n$ geometric objects, subject to insertions and deletions of objects, so that we can quickly determine whether two query objects are connected in the intersection graph.
\end{quote}
The challenge here is that a single insertion/deletion of an object may change $\Omega(n)$ edges in the intersection graph in the worst case, and so we can't afford to maintain the intersection graph explicitly.

\subparagraph*{Previous work by Chan.} Chan~\cite{Chan06} obtained the first fully dynamic data structure with sublinear ($O(n^{0.94})$) amortized update time and sublinear ($\OO(m^{1/3})$)%
\footnote{Throughout the paper, the $\OO$ notation hides $\log^{O(1)}n$ factors, and the $\OOO$ notation hides $n^\eps$ factors for an arbitrarily small constant $\eps>0$.} 
query time for axis-aligned line segments or rectangles in 2D
or axis-aligned boxes in any constant dimension.  The exponent ($0.94$) arose from 
matrix multiplication.  The heart of his solution was a data structure
for an extension of dynamic graph connectivity with vertex (re)insertions and deletions,
called \emph{dynamic subgraph connectivity}---maintaining a (sparse) graph with $m$ edges and a subset $A$ of ``active'' vertices under
(re)insertions and deletions in $A$, so that we can quickly determine whether two vertices are connected in
the subgraph induced by $A$.  As Chan observed, the geometric problem in the case of axis-aligned
boxes can be reduced to dynamic subgraph connectivity (ignoring polylogarithmic factors), by using so-called 
``biclique covers'' (which are related to range searching or intersection searching).  

For more general classes of geometric objects such as arbitrary line segments, biclique covers
have superlinear complexity, so Chan's reduction was not sufficient to yield sublinear update time, unfortunately.

\subparagraph*{Previous work by Chan, \Patrascu, and Roditty.}
A subsequent paper by Chan, \Patrascu, and Roditty~\cite{ChanPR11} (FOCS'08) rectified the problem:
They first obtained a better data structure for dynamic subgraph connectivity, avoiding
fast matrix multiplication, and then incorporated range searching techniques into their data
structure in a more efficient way, thereby obtaining data structures for dynamic geometric connectivity
with sublinear amortized update time for virtually all families of objects with constant description complexity.
For example, for axis-aligned line segments or rectangles in 2D and axis-aligned boxes in
any constant dimension, the amortized update time is $\OO(n^{2/3})$
and query time is $\OO(n^{1/3})$;
for arbitrary line segments in 2D, the amortized update time is $\OOO(n^{9/10})$ and
query time is $\OO(n^{1/5})$;
for disks in 2D, the amortized update time is $\OOO(n^{20/21})$ and query time
is $\OOO(n^{1/7})$.

Recently, Jin and Xu~\cite{JinX22} proved an $\Omega(m^{2/3-\eps})$ conditional lower bound
on the amortized update time for the dynamic subgraph connectivity problem for any data structure
with $O(m^{1-\eps})$
query time,\footnote{Throughout the paper, $\eps$ denotes
an arbitrarily small positive constant.}
assuming the ``Combinatorial 4-Clique Hypothesis'', if we restrict
ourselves to ``combinatorial'' algorithms that do not
use algebraic techniques for fast matrix multiplication.  Chan~\cite{Chan06} observed that dynamic subgraph connectivity
with $m$ edges can be reduced back to dynamic geometric connectivity for $O(m)$ axis-aligned line segments or boxes in 3D (roughly speaking, because in 3D, axis-aligned segments or boxes can ``simulate'' arbitrary graphs), and so
Jin and Xu's proof implies conditional optimality of Chan, \Patrascu, and Roditty's $\OO(n^{2/3})$ update bound for axis-aligned boxes
in dimension 3 and above, at least with respect to combinatorial algorithms.

However, one major drawback in Chan, \Patrascu, and Roditty's data structures is that the query times are super-polylogarithmic
(unlike known static data structures).  One can envision applications where queries occur much more frequently
than updates.  Although trade-offs with smaller query time and larger update time are
possible with Chan, \Patrascu, and Roditty's data structures, they do not achieve sublinear update
time in the regime of constant or polylogarithmic query time. 

In some ways, their data structures
encode connectivity information only ``implicitly''.  Thus, another drawback is that they cannot handle ``global'' connectivity
queries, e.g., deciding whether the entire intersection graph is connected, or counting the
number of connected components.  Abboud and Vassilevska Williams~\cite{AbboudW14} proved an $\Omega(m^{1-\eps})$
conditional lower bound on the query/update time for global connectivity
queries for the dynamic subgraph connectivity problem, assuming the Strong Exponential-Time Hypothesis.  So, 
sublinear query and update time are conditionally not possible for global connectivity queries for
axis-aligned segments and boxes in dimension 3 and above.
Still, a fundamental question remains as to whether sublinear update time bound is achievable for global connectivity queries for various types of geometric objects in 2D (which is where the most natural geometric applications occur).

\subparagraph*{Previous work by Afshani and Chan.}
A different (earlier) data structure by Afshani and Chan~\cite{AfshaniC09} (ESA'06) addressed some of these drawbacks.
Specifically, they obtained $O(1)$ query time and $\OO(n^{10/11})$ update time in the case of
axis-aligned line segments or rectangles in 2D\@.  Unfortunately, they were unable to extend their method to
other types of objects in 2D, such as arbitrary line segments or disks.

\subparagraph*{Other related work.}
Better data structures are known for some easier special cases; for example, a recent SoCG'22 paper
by Kaplan et al.~\cite{KaplanKKKMRS22} gave polylogarithmic results for unit disks, or
disks with small maximum-to-minimum-radius ratio,
or for arbitrary disks in the incremental (insertion-only) or decremental (deletion-only) case (see also Kaplan et al.'s SOSA'24 paper~\cite{KaplanKKMR24}). 
However, these results are not applicable to arbitrary disks in the fully dynamic setting.

\subparagraph*{New results.}
We obtain new data structures for dynamic geometric connectivity for different types of objects in 2D,
as summarized in Table~\ref{tab:results}.%
\footnote{For the previous result on simplices in $\R^d$, Chan, \Patrascu, and Roditty~\cite{ChanPR11} originally stated query time $\OOO(n^{1/(2d+1)})$ and update time $\OOO(n^{1-1/d(2d+1)})$, but they mistakenly assumed bounds for simplex range searching extend to simplex intersection searching. However, simplex intersection searching in $\R^d$ does reduce to semialgebraic range searching
in $\R^{O(d^2)}$, so their framework implies the sublinear bounds shown in the table.}

Qualitatively, we obtain the first data structures with \emph{constant} query time and sublinear
amortized update time for arbitrary line segments in 2D as well as disks in 2D, and the first data structures
that can handle global connectivity queries for such objects.  The line segment
case is especially general (since it can handle arbitrary  polylines of bounded complexity).  Our approach
can in fact handle any family of semialgebraic curves with constant description complexity in 2D, though
the exponent in the update bound depends on the degree of the curves.

Quantitatively, for the case of disks in 2D, our data structure not only greatly reduces the query time
of Chan, \Patrascu, and Roditty's method but does so without hurting the update time---in fact, the 
update bound is \emph{improved} (albeit slightly), from near $n^{20/21}$ to $n^{7/8}$.
For the special case of axis-aligned line segments in 2D, our data structure also improves
the update time of Afshani and Chan's method, from near $n^{10/11}$ to $n^{4/5}$.

\begin{table}[]
\renewcommand{\arraystretch}{1.1}
    \centering
    \begin{tabular}{|l|ll|l|l|}
        \hline
        Type of objects  & Query time & Update time (amort.) & Ref.\\
        \hline
        Axis-aligned segments/boxes in $\R^d$ & $\Tilde{O}(n^{1/3})$ & $\Tilde{O}(n^{2/3})$ &   \\
        Arbitrary line segments in $\R^2$ & $O^*(n^{1/5})$ & $O^*(n^{9/10})$ & \\
        Disks in $\R^2$ & $O^*(n^{1/7})$ & $O^*(n^{20/21})$ &  \multirow{4}{*}{\cite{ChanPR11}}\\
        Balls in $\R^d$ & $O^*(n^{\frac{1}{2d+3}})$ & $O^*(n^{1-\frac{1}{(d+1)(2d+3)}})$ & \\
        Simplices in $\R^d$ & $O^*(n^{\frac{1}{\Theta(d^2)}})$ & $O^*(n^{1-\frac{1}{\Theta(d^4)}})$ & \\
        \hline
        Axis-aligned segments/rectangles in $\R^2$ & $O(1)$ & $\Tilde{O}(n^{10/11})$ & \cite{AfshaniC09} \\
        \hline
        Axis-aligned segments in $\R^2$ & $O(1)$ & $\tildeO(n^{4/5})$ &  \multirow{3}{*}{new} \\
        Arbitrary line segments in $\R^2$ & $O(1)$ & $\OOO(n^{20/21})$ & \\
        Disks in $\R^2$ & $O(1)$ & $O^*(n^{7/8})$ & \\
        \hline
    \end{tabular}
    \vspace{1ex}
    \caption{Previous and new  results on dynamic geometric connectivity.}
    \label{tab:results}
\end{table}


\subparagraph*{Techniques and new ideas.}
Our approach builds on Afshani and Chan's method~\cite{AfshaniC09}, based on the key notion of 
\emph{equivalence classes} of connected components.  They proved combinatorial bounds
and described efficient algorithms for computing and maintaining such classes.  Although
their combinatorial bounds actually hold for arbitrary geometric objects in 2D, their algorithms are
specialized to \emph{axis-aligned} objects in 2D.  We describe a different way to
compute and maintain equivalence classes of components, based on a \emph{repeated splitting}
idea which is simple in hindsight, but has somehow been overlooked by researchers for more than 
a decade.  This idea allows us to reduce equivalence-class data structures 
to \emph{colored} variants of range or intersection searching, which can be solved by known geometric data structuring techniques.
This new idea 
allows us to obtain constant query time and sublinear amortized
update time for all the types of 2D geometric objects considered here; see Sections~\ref{sec:key}--\ref{sec:dynconn}.

Of independent interest is also a new variant of Afshani and Chan's combinatorial lemma, which helps in reducing the update time further; see Section~\ref{sec:new-comb-lemma}.

In Appendix~\ref{apd:sep}, we obtain a still further improvement for the case of disks
by incorporating another new idea---namely, the usage of \emph{separators} (specifically, Smith and Wormald's geometric
 separator theorem~\cite{SmithW98}).
 This is interesting, as separators have not been widely used before in the context of dynamic geometric connectivity (they have been used in dynamic subgraph connectivity for planar graphs~\cite{FrigioniI00} but not for more general classes of geometric intersection graphs).



\section{Recap of Afshani and Chan's Method}
\label{sec:recap}

Before describing our new data structures, we first review Afshani and Chan's previous method for dynamic connectivity for axis-aligned segments/rectangles in 2D~\cite{AfshaniC09}. To illustrate the overall ideas, for simplicity we will focus on the \emph{offline} setting of the problem, where the upcoming updates are known in advance (the online setting requires more work and slightly worse update time),
and we will focus on just the case of axis-aligned segments.


As often seen in previous works~\cite{Chan03,Chan06}, Afshani and Chan used the standard technique of handling insertions lazily and rebuilding periodically. Each period phase consists of a preprocessing step followed by $q$ updates, so that the preprocessing cost can be amortized. In the offline setting, we know which objects will be updated in each phase.

Let $S$ be the objects that exist at the beginning of the phase and will not be deleted during the phase, and let $Q$ be the sequence of objects $s_1,\ldots,s_q$ that will be updated during the phase. (Thus, $S$ stays static during a phase, whereas $Q$ is small.)  Afshani and Chan defined two connected components of $S$ to be {\em equivalent} if they intersect the exact same set of objects from $Q$. Clearly, for any $s\in Q$, if $s$ intersects an arbitrary component from a class $L$, then $s$ intersects all components from $L$. More importantly, using the fact that distinct connected components never intersect each other and may be viewed as a collection of disjoint ``curves'' or ``strings'' (not necessarily of constant complexity) when the objects are in $\R^2$, Afshani and Chan proved a polynomial upper bound on the number of equivalence classes at any moment. The general combinatorial lemma is stated as follows.

\begin{lemma}[Afshani and Chan's combinatorial lemma~\cite{AfshaniC09}]
    \label{lem:q3-classes}
    Consider a set $Q$ of $q$ disjoint regions with simple connected boundaries and a set $C$ of disjoint curves in $\R^2$. Then $C$ consists of at most $O(q^3)$ equivalence classes with respect to $Q$.  (This bound is tight.)
\end{lemma}



To readers familiar with the notion of ``VC dimension'' or ``shatter dimension'' of set systems~\cite{MatousekBook},
the above lemma is equivalent to the statement that
the set system $(Q,\calR)$ with  $\calR=\{\{s\in Q: s\cap c\neq \emptyset\}: c\in C\}$ has shatter dimension at most 3.
We stress that what makes the lemma interesting is that the curves in $C$ do not need to have constant complexity.  On the other hand, disjointness of the curves in $C$ is crucial (otherwise, it is easy to find counterexamples with exponentially many
equivalence classes).

In the dynamic connectivity problem, the inserted segments are not disjoint. However, the arrangement of $Q$ can be broken down into $O(q^2)$ non-intersecting sub-segments.
This implies an $O(q^6)$ bound on the number of equivalence classes.

Below, we sketch how equivalence classes can be used to obtain an offline dynamic connectivity data structure for axis-aligned segments.

\subparagraph*{Preprocessing.} At the beginning of a phase, the first task is to find the set of equivalence classes among components of $S$ with respect to $Q$.  This step is nontrivial, but Afshani and Chan showed how to compute these classes, or more precisely, the equivalences classes with respect to the $O(q^2)$ non-intersecting sub-segments, in $\OO(n)$ time in the case of axis-aligned segments  (we will say more about this later).

Throughout the phase, we maintain a {\em proxy graph} $H$, such that the connected components of the geometric objects roughly correspond to the connected components of $H$. The graph $H$ is defined as follows:

\begin{itemize}
    \item For each equivalence class $L$, there is a \emph{class vertex} corresponding to $L$.
    \item For each inserted segment $s\in Q$, there is an \emph{insertion vertex} corresponding to $s$.
    \item The edges of $H$ indicate which pairs of objects intersect.
\end{itemize}

The graph $H$ has $O(q^6)$ vertices and $O(q^7)$ edges, and can be stored in a dynamic graph connectivity structure supporting edge updates in $\OO(1)$ time~\cite{HolmLT01,HuangHKPT23}.

\subparagraph*{Offline updates.} Updates of objects are done only to $Q$ and not to $S$.  Whenever an object $s$ is inserted to $Q$, we create a vertex for $s$ in $H$. For each class $L$, we decide whether $s$ intersects all components of $L$ by picking an arbitrary component from $L$, and then query an orthogonal intersection searching structure~\cite{EdelsbrunnerM81}.
Each insertion/deletion in $Q$ causes
$\OO(q^6)$ edge updates in $H$.

The preprocessing step takes $\tildeO(n)$ time, while each update takes $\tildeO(q^6)$ time. Choosing $q=n^{1/7}$, we achieve $\tildeO(n^{6/7})$ amortized update time. 

\subparagraph*{Query.} Given two query objects $u,v$, there are now two cases to consider.

\begin{itemize}
    \item If both $u,v\in S$, then let $c_u,c_v$ be the components containing $u,v$. If $c_u=c_v$, then $u$ and $v$ are connected. If $c_u\neq c_v$ and either $c_u$ or $c_v$ belongs to an equivalence class that corresponds to an isolated vertex in $H$, then $u$ and $v$ are not connected.
    \item Otherwise, $u$ and $v$ are connected if and only if their corresponding vertices in $H$ are connected (if $u\in S$, then its corresponding vertex is the class containing $c_u$).
\end{itemize}

Thus, queries between two objects can be handled in $O(1)$ time.  We can also handle global connectivity queries: the overall number of connected components is equal to the number of components in $H$ (which can be maintained by a dynamic graph connectivity structure~\cite{PatrascuT07}) plus the number of components in isolated vertices of $H$ (which is straightforward to maintain).

\section{A New Version of the Combinatorial Lemma}
\label{sec:new-comb-lemma}

It turns out that the $O(q^6)$ combinatorial bound on the number of equivalence classes, when $Q$ may be intersecting,  can be improved to $O(q^3)$.  There is no need to subdivide the arrangement of $Q$ into $O(q^2)$ non-intersecting parts before applying Lemma~\ref{lem:q3-classes}. Rather, we can modify Afshani and Chan's proof directly.

\begin{lemma}[New combinatorial lemma]
    \label{lem:q3-classes-objects}
    Consider a set $Q$ of $q$ regions with simple connected boundaries and a set $C$ of disjoint curves in $\R^2$, where the boundaries of any two regions from $Q$ intersect at most $O(1)$ times. Then $C$ consists of at most $O(q^3)$ equivalence classes with respect to $Q$.  (This bound is tight.)
\end{lemma}

\begin{proof}
We modify Afshani and Chan's proof~\cite{AfshaniC09}.
    For each curve $c\in C$, we first replace $c$ with a minimal sub-curve that intersects the same set of regions as $c$.  (In other words, we are ``erasing'' unnecessary parts of the input curves, which do not change the equivalence classes; see Figure~\ref{fig:erase}.) Then the two endpoints of $c$ must lie on the boundaries of two different regions in $Q$ (except for curves that intersect only one region, but there are only $O(q)$ classes of such curves).

    \begin{figure}
        \centering
        \begin{subfigure}{.5\textwidth}
          \centering
          \includegraphics[width=.45\linewidth,page=1]{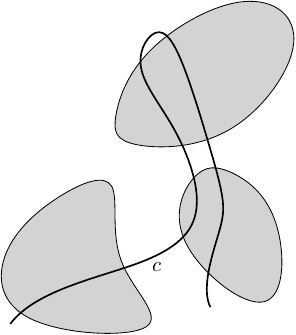}
          \label{fig:erase-1}
        \end{subfigure}%
        \begin{subfigure}{.5\textwidth}
          \centering
          \includegraphics[width=.45\linewidth,page=2]{figures/comb-lemma.pdf}
          \label{fig:erase-2}
        \end{subfigure}
        \caption{The left figure shows a curve $c$ intersecting three regions in $Q$. The right figure shows the minimal sub-curve of $c$ that intersects the same regions as $c$.}
        \label{fig:erase}
    \end{figure}
    
    Fix two different regions $r,r'\in Q$.  Divide $\partial r$ and $\partial r'$ into $O(1)$ pieces at their intersection points, if they intersect.  Take a piece $\pi$ of $\partial r$ and a piece $\pi'$ of $\partial r'$.  We will bound the number of equivalence classes of the subset $C_{\pi,\pi'}$ of all curves in $C$ that have one endpoint on $\pi$ and one endpoint on $\pi'$.

    \begin{figure}
        \centering
        \includegraphics[width=.45\textwidth,page=3]{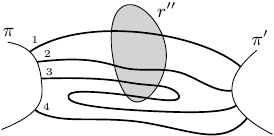}
        \caption{This figure shows four curves in $C_{\pi,\pi'}$, ordered from $1$ to $4$ in cyclic order. Any region $r''$ disjoint from $\pi$ and $\pi'$ must intersect a contiguous (cyclic) sequence of the curves.}
        \label{fig:curve-order}
    \end{figure}
    
    The curves in $C_{\pi,\pi'}$ do not intersect each other, and they also do not intersect $\pi$ and $\pi'$, because of minimality (otherwise we could have erased more).  Thus, we can sort all these curves along $\pi$, or cyclically equivalently, along $\pi'$.  
    For any region $r''$ whose boundary does not intersect $\pi$ and $\pi'$,
    the curves in $C_{\pi,\pi'}$ that intersect $r''$ are contiguous in the sorted cyclic sequence (see Figure~\ref{fig:curve-order}).
    More generally, for an arbitrary region $r''\in Q$, we can divide into $O(1)$ sub-regions with boundaries not intersecting $\pi$ and $\pi'$.
    Thus, for any region $r''\in Q$, the curves in $C_{\pi,\pi'}$ that intersect $r''$ form a union of $O(1)$
    contiguous subsequences in the sorted sequence.
    Taking all the $O(q)$ delimiters together, we get a partition of $C_{\pi,\pi'}$ into $O(q)$ blocks such that two curves in the
    same block are in the same equivalence class.  Hence, the number of equivalence classes in $C_{\pi,\pi'}$ is $O(q)$.
    
    As there are $O(q^2)$ choices of $(\pi,\pi')$, the total number of equivalence classes is at most $O(q^3)$.
\end{proof}

Note that the lemma holds also when $C$ is a set of disjoint connected regions (since regions may be ``simulated'' by strings~\cite{Lee17}).

\section{Solving the Key Subproblem: Computing Equivalence Classes}
\label{sec:key}

As illustrated by Afshani and Chan's method, dynamic connectivity---at least in the offline settings---reduces to the following key subproblem.

\begin{subproblem}[Equivalence class computation]
    \label{prb:compute-eq-cls}
    Given a set $S$ of $n$ geometric objects and a set $Q$ of $q$ geometric objects in $\R^2$, compute the equivalence classes among connected components of $S$ with respect to $Q$ faster than $O(qn)$ time (ideally, in $\OO(n)$ time when $q$ is not too large).
\end{subproblem}

For axis-aligned line segments, Afshani and Chan solved Subproblem~\ref{prb:compute-eq-cls} roughly as follows.
We first form a grid by drawing the grid lines at the endpoints of the segments of $Q$. The grid lines then divide the segments of $Q$ into $O(q^2)$ non-intersecting sub-segments. We will actually
compute equivalence classes with respect to these $O(q^2)$ sub-segments instead of $Q$ (this is sufficient for the application to dynamic connectivity).
The sub-segments within each row (or column) are linearly ordered. Since any axis-aligned segment $u$ lies in exactly one row or column, the set of sub-segments intersecting $u$ can be represented as an integer interval, which can be computed in $\tildeO(1)$ time using binary search. For each connected component $c$ of $S$, the class that $c$ belongs to can be represented as the union of the intervals for all segments of $c$. Thus, the representations of all components can be found in $\tildeO(n)$ time. Crucially, two components have the same representation if and only if they belong to the same equivalence class. With some careful analysis, Afshani and Chan showed that the representations can be sorted in $\tildeO(n)$ time, allowing us to then compute all equivalence classes by a linear scan in $\tildeO(n)$ total time.
 
However, as one can see, the above ad hoc, grid-based approach does not work when the
objects are not axis-aligned (it fails even for line segments with three possible slopes).

In this section, we propose a more general approach to solving Subproblem~\ref{prb:compute-eq-cls} that can handle most types of objects in 2D, including disks and arbitrary line segments.
Our approach does not require subdividing $Q$ into nonintersecting pieces first, so we can exploit our new improved combinatorial lemma to get better results even for axis-aligned segments.
Furthermore, the \emph{incremental} manner in which we compute the equivalence classes will help in handling online updates more efficiently in our dynamic connectivity data structures.

\subsection{New Idea: Repeated Class Splitting}

Our central idea is simple: we will compute the equivalence classes incrementally by inserting objects of $Q$ one at a time. (This is fundamentally different from Afshani and Chan's approach, which scans through elements of $S$ instead.)  Initially, no object has been inserted, so all components of $S$ belong to the same class. Whenever an object $s$ is inserted, we examine each class $L$ that existed before the insertion. If $s$ intersects some but not all components in $L$, then $L$ must be split into two ``child classes''---one whose components intersect $s$, and one whose components don't. This way, we can maintain the equivalence classes by repeatedly splitting existing classes as objects are being inserted.  (A similar idea was used by Chan~\cite{Chan21} to solve a completely different problem.)

\subparagraph*{Class splitting data structure.} To split a class into two child classes, we design a {\em class splitting} data structure $\calD(L)$, such that for any object $s$ that splits $L$ into child classes $L_1,L_2$, we can efficiently produce data structures $\calD(L_1)$ and $\calD(L_2)$. An ostensible obstacle is that when forming $\calD(L_1)$ and $\calD(L_2)$, it seems that we can't avoid spending time at least $O(1)$ time per object in $L$. If so, computing the classes may take $\Omega(n)$ time per insertion.

To get around this obstacle, we use an amortization trick, namely, a reverse version of the standard ``weighted union heuristic'':
instead of building both $\calD(L_1)$ and $\calD(L_2)$ from scratch, we do so only for the child class with smaller total size. Specifically, if $L_2$ has smaller total size than $L_1$, then we build $\calD(L_2)$ from scratch, while $\calD(L_1)$ is obtained by deleting everything associated with $L_2$ from $\calD(L)$. Intuitively, we ``displace'' $L_2$ from $L$ without necessarily examining objects from $L_1$. 
We claim that the sum of the number of displaced objects over all updates in the phase
is $O(n\log n)$.  One proof is via a potential function argument: defining $\Phi=\sum_{\mbox{\scriptsize\rm  class $L$}} |L|\log |L|$ (where $|L| := \sum_{\mbox{\scriptsize component\ }c\in L}|c|$), it is not difficult to
see that a split with $y$ displacements decreases $\Phi$ by at least $\Omega(y)$ (because of the inequality
$(x+y)\log(x+y) - x\log x - y\log y\ge y\log (x/y+1)\ge y$ when $x\ge y$), and the claim follows.
(An alternative, simple proof is to argue that each object can be displaced at most $O(\log n)$ times, since each time this happens, the size of the class containing the object is reduced by at least a factor of 2.)

\subparagraph*{Reduction to component (non-)intersection reporting.} The key part of the class splitting operation is building a data structure $\calD(L)$ that can report the smaller child class $L'$ in time linear in the total size of $L'$, but sublinear in the total size of $L$. It suffices for the data structure to handle the following two types of queries:

\begin{enumerate}
    \item Report all connected components in $L$ that intersect $s$.
    \item Report all connected components in $L$ that {\em do not} intersect $s$.
\end{enumerate}

By running the two query-answering algorithms concurrently and stopping when either of them terminates, we can report $L'$ within the desired time bound.  We need the data structure to support deletions of components, so that we can subsequently delete each component in $L'$ from~$L$.


The component intersection reporting problem we face can be viewed as a \emph{colored} intersection reporting problem, where the objective is to preprocess a set of colored objects, so that one can efficiently report all colors intersected by a query object. 
Similarly, we need to solve the corresponding
colored non-intersection reporting problem.
(See \cite{GuptaJRS18} for background on colored range searching.)

In the next three subsections, we apply known geometric data structuring techniques (some of which are 
inspired by colored range searching), to solve this component (non-)intersection reporting problem
for axis-aligned line segments, disks, and arbitrary line segments.

\subsection{Axis-Aligned Line Segments}

\begin{lemma}
    \label{lem:seg0-comp-search}
    Given a set of connected components formed by $n$ axis-aligned segments in $\R^2$,
    there is a component (non-)intersection reporting data structure with $\OO(n)$ preprocessing time and $\tildeO(1+k)$ query time, where $k$ is the number of components reported.  Furthermore, we can support deletion of any component in $\OO(1)$ amortized time, and insertion of any component $c$ in $\OO(|c|)$ amortized time.
\end{lemma}

\begin{proof}
    In the following, we assume that the inserted segment is vertical; the other case can be handled symmetrically. For each component $c$, we compute the vertical decomposition $T_c$ of the horizontal segments of $c$. There are $O(n)$ rectangular cells over all components. We make the following observation.
    
    \begin{observation}
        \label{obs:seg0-cells}
        Let $s$ be a vertical segment and $p$ be its lower endpoint. For any component $c$, let $\Delta$ be the cell in $T_c$ that contains $p$, and let $u$ be the horizontal segment in $c$ that bounds $\Delta$ from above. Then $s$ intersects a segment in $c$ if and only if $s$ intersects $u$.
    \end{observation}
    
    By Observation~\ref{obs:seg0-cells}, components that intersect $s$ can be reported as follows:
    
    \begin{enumerate}
        \item Among the cells of $T_c$ for all components $c$, find all cells that contain $p$.
        \item Report, among the upper-bounding segments of these cells, the ones that intersect $s$.
    \end{enumerate}
    
    Components that do not intersect $s$ can be reported similarly, except in the second step we instead report the segments that do not intersect $s$. We support the above reporting query with a two-level data structure $\calD$.
    
    \begin{itemize}
        \item The primary structure is an orthogonal rectangle stabbing structure with $\tildeO(1)$ query time and $\tildeO(1)$ update time (e.g., \cite{EdelsbrunnerM81}, or by reduction to orthogonal range searching in~$\R^4$).
        \item For each canonical subset $B$ in the primary structure, we associate each cell $\Delta\in B$ with the segment that bounds $\Delta$ from above. We keep two auxiliary structures on the associated line segments. Given a vertical segment $s$, one auxiliary structure reports all segments intersecting $s$, while the other reports all segments not intersecting $s$. Both structures have $\tildeO(1+k)$ query time (where $k$ is the output size) and $\tildeO(1)$ update time (as these subproblems also reduce to orthogonal range searching).
    \end{itemize}

    The preprocessing and query time bounds are clearly satisfied. Insertion/deletion of a component $c$ requires inserting/deleting all cells of $T_c$ and  takes $\OO(|c|)$ time; we can charge the cost of deletion to preprocessing or insertion.
\end{proof}

\begin{corollary}
    \label{cor:seg0-split-time}
    For $n$ axis-aligned line segments in $\R^2$, we can solve Subproblem~\ref{prb:compute-eq-cls} in $\OO(n+q^4)$ time. Furthermore, each insertion to $Q$ takes  $\tildeO(q^3 + n/q)$ amortized time. We can also support deletion of a component in $\OO(1)$ amortized time and  insertion of a component $c$ in $\OO(|c|)$ amortized time, if we are given $c$'s class.
\end{corollary}

\begin{proof}
    We store each equivalence class in the data structure from the preceding lemma.
    For $i=1,\ldots,q$, let $L_{i,1},\ldots,L_{i,O(q^3)}$ be the existing classes before the $i$-th insertion. Let $n_{i,1},\ldots,n_{i,O(q^3)}$ be the sizes of the classes. Whenever a class $L_{i,j}$ is split, let $m_{i,j}$ be the total size of the displaced child class. As noted, the total number of displacements is $\sum_{i,j} m_{i,j} = \tildeO(n)$, when there are no insertions of components; in general, it is $\OO(n+\Sigma)$, where $\Sigma$ is the sum of the sizes of the components inserted (because insertion of a component $c$ increases the potential $\Phi$ by $\OO(|c|)$). Thus, the total running time for class splitting operations over all insertions to $Q$ is at most
    \begin{equation*}
        \tildeO\left( \sum_{i=1}^{q} \sum_{j=1}^{O(q^3)} \left(1 + m_{i,j}\right) \right)
        \:=\: \tildeO\left(q^4 + n + \Sigma\right).
    \end{equation*}
    This amortizes to $\tildeO(q^3 + n/q)$ time per insertion to $Q$, with the $\Sigma$ term charged to insertions of components.
\end{proof}

\subsection{Disks}

\begin{lemma}
    \label{lem:disk-comp-search}
    Given a set of connected components formed by $n$ disks in $\R^2$, there is a component (non-)intersection reporting data structure with $\OOO(n)$ preprocessing time and $\OOO(n^{4/5}+k)$ query time, where $k$ is the number of components reported. Furthermore, we can support deletion of a component in $\OOO(1)$ amortized time and insertion of a component $c$ in $\OOO(|c|)$ amortized time.
\end{lemma}

The proof follows the same argument as for axis-aligned segments, except we use different data structures than orthogonal range searching for the primary and auxiliary structures.

\begin{proof}
    The component (non-)intersection searching structure for disks is similar to the one for axis-aligned segments. However, for each component $c$, instead of using vertical decomposition, we compute an additively weighted Voronoi diagram~\cite{LeeD81} with the disk centers as the sites and the radii as the weights.  Let $V_c$ be a decomposition of the diagram into cells of constant complexity, e.g., by radially decomposing each Voronoi cell with its site as the origin.  Then $V_c$ satisfies the same key property as the vertical decomposition in the case of axis-aligned segments.
    
    \begin{observation}
        \label{obs:disk-cells}
        Let $s$ be a disk and $p$ be its center. For any component $c$, let $\Delta$ be the cell in $V_c$ that contains $p$, and let $u$ be the disk associated with $\Delta$. Then $s$ intersects a disk in $c$ if and only if $s$ intersects $u$.
    \end{observation}
    
    By Observation~\ref{obs:disk-cells}, components that intersect $s$ can be reported as follows:
    
    \begin{enumerate}
        \item Among the cells of $V_c$ for all components $c$, find all cells that contain $p$.
        \item Report, among the disks associated with these cells, the ones that intersect $s$.
    \end{enumerate}
    
    Components that do not intersect $s$ can be reported similarly, except in the second step we instead report the disks that do not intersect $s$. Again, we support the above reporting query with a two-level data structure $\calD$.
    
    \begin{itemize}
        \item For the primary structure, observe that 
        edges of additively weighted Voronoi diagrams in $\R^2$ are hyperbolas, which have equations of the form $c_1x^2+c_2y^2+c_3 xy + c_4x+c_5y=1$ and can be ``linearized'' via the mapping $(x,y)\mapsto (x^2,y^2,xy,x,y)$. Thus, stabbing queries on the cells of the $V_c$'s reduce to simplex-stabbing queries in $\R^5$. We can handle such queries using (a multi-level version of) \Matousek{}'s partition tree~\cite{Matousek92} in $\R^5$,
        which has $\OOO(n)$ preprocessing time and $O^*(1)$ update time, such that
        given any query point $p$, we can in $O^*(n^{4/5})$ time find $O^*(n^{4/5})$ canonical subsets of the simplices containing $p$.
        
        \item For each canonical subset $B$ in the primary structure, let $S_B$ be the set of disks associated with the simplices in $B$. We keep two auxiliary structures on the associated disks. Given an input disk $s$, one auxiliary structure reports all disks that intersect $s$, while the other reports all disks that do not intersect $s$. Both structures reduce to repeated nearest/farthest neighbor queries for points with additive weights in $\R^2$. We can use results by Agarwal, Efrat, and Sharir~\cite{AgarwalES99} to answer queries in $\OOO(1+k)$ time with $\OOO(1)$ update time (or alternatively, Kaplan et al.'s data structure~\cite{KaplanMRSS20} with slightly better, polylogarithmic bounds).
    \end{itemize}

    The two-level data structure clearly satisfies the desired time bounds.
\end{proof}

\begin{corollary}
    \label{cor:disk-split-time}
    For $n$ disks in $\R^2$, we can solve Subproblem~\ref{prb:compute-eq-cls} in $\OOO(n+q^{8/5} n^{4/5})$ time. Furthermore, each insertion to $Q$ takes amortized $\OOO(q^{3/5} n^{4/5} + n/q)$ time. We can also support deletion of a component in amortized $\OOO(1)$ time and insertion of a component $c$ in amortized $\OOO(|c|)$ time, if we are given $c$'s class.
\end{corollary}

\begin{proof}
    We define $n_{i,j}$ and $m_{i,j}$ for $i=1,\ldots,q$ and $j=1,\ldots,O(q^3)$, as well as $\Sigma$, in the same way as we did for axis-aligned segments. Using the same analysis, the total running time for class splitting operations over all insertions to $Q$ is at most
    \begin{equation*}
        O^*\left( \sum_{i=1}^{q} \sum_{j=1}^{O(q^3)} \left(n_{i,j}^{4/5} + m_{i,j} \right) \right),
    \end{equation*}
    which, by \Holder's inequality, is bounded by $O^*(q\cdot (q^{3/5} n^{4/5}) + n + \Sigma)$. This amortizes to $O^*(q^{3/5} n^{4/5} + n/q)$ time per insertion to $Q$.
\end{proof}

\subsection{Arbitrary Line Segments}

The component intersection reporting problem becomes tougher for the case of arbitrary line segments.  Although we are unable to get near-linear preprocessing time, we can obtain the following preprocessing/query-time trade-off.  The idea uses \emph{cuttings} and is inspired by a known idea for a different colored range searching problem (namely, ``range mode'', i.e., finding the most frequent color in a range)~\cite{ChanDLMW14}.

\begin{lemma}\label{lem:seg:cuttings}
    Given a set of connected components formed by $n$ line segments in $\R^2$ and a parameter $r$, there is a component (non-)intersection reporting data structure with $\OOO(nr^4)$ preprocessing time and $\OOO(1 + n/r + k)$ query time, where $k$ is the number of components reported. Furthermore, we can support deletion of a component in $\OOO(r^4)$ amortized time and insertion of a component $c$ in  $\OOO(|c|r^4)$ amortized time.
\end{lemma}

\begin{proof}
    A line segment in $\R^2$ can be represented by 4 real values (the coordinates of its two endpoints) and so can be viewed as a point in $\R^4$.
    For each input line segment $s$, the set of all line segments intersecting $s$ then corresponds a semialgebraic set in $\R^4$ with constant description complexity.  Let $S$ be the collection of these $n$ semialgebraic sets.
    
    Now, compute a \emph{$(1/r)$-cutting} $\Gamma$ of $S$, consisting of
    $\OOO(r^4)$ cells (due to known combinatorial bounds on vertical decompositions)~\cite{Koltun04},
    in $\OOO(nr^4)$ time.
    
    For each cell $\Delta\in\Gamma$, we store the \emph{conflict list} of $\Delta$, consisting of all semialgebraic sets $s\in S$ with boundaries crossing $\Delta$; by definition of cuttings, each conflict list has size $O(n/r)$. For each component $c$ with no semialgebraic set in the conflict list of $\Delta$, we store $c$ in a list $A(\Delta)$ if at least one semialgebraic set from $c$ completely contains $\Delta$, or store $c$ in another list $B(\Delta)$ otherwise. To allow for efficient update, for each component $c$, we keep track of all semialgebraic sets from $c$ that contains $\Delta$.
    
    For each semialgebraic set $s\in S$ and each cell $\Delta\in\Gamma$, we can decide in $O(1)$ time which list $s$ should be stored in. Thus, the data structure has $O^*(nr^4)$ preprocessing time. For the same reason, deletion of a semialgebraic set can be handled in $O^*(r^4)$ time.
    
    To handle a query for a line segment represented as a point $p\in\R^4$, we first locate the cell $\Delta\in\Gamma$ containing $p$. To report all components stabbed (resp. not stabbed) by $p$, we linearly search the conflict list, and then report the components stored in the list $A(\Delta)$ (resp. $B(\Delta)$). The reporting takes $O^*(1 + n/r + k)$ time, where $k$ is the output size.

    Insertions of components can be handled by the logarithmic method~\cite{BentleyS80}.
\end{proof}

\begin{corollary}
    \label{cor:segs-split-time}
    For $n$ line segments in $\R^2$ and a parameter $r$, we can solve Subproblem~\ref{prb:compute-eq-cls} in $\OOO(r^4n + qn/r + q^4)$ time. Furthermore, each insertion to $Q$ takes amortized $\OOO(n/r + q^3 + r^4 n/q)$ time. We can also support deletion of a component in $\OOO(r^4)$ time and insertion of a component $c$ in amortized $\OOO(|c|r^4)$ time, if we are given $c$'s class.
\end{corollary}

\begin{proof}
    We define $n_{i,j}$ and $m_{i,j}$ for $i=1,\ldots,q$ and $j=1,\ldots,O(q^3)$, as well as $\Sigma$, in the same way as we did for axis-aligned segments and disks. Using the same analysis, the total running time for class splitting over all insertions to $Q$ is at most
    \begin{equation*}
        \OOO\left( \sum_{i=1}^{q} \sum_{j=1}^{O(q^3)} \left( 1 + \frac{n_{i,j}}{r} + m_{i,j} r^{4} \right) \right)
        \:=\: \OOO\left( \frac{qn}{r} + q^4 + r^4 (n + \Sigma) \right),
    \end{equation*}
    This amortizes to $\OOO(n/r + q^3 + r^4 n/q)$ per insertion to $Q$.
\end{proof}

\section{Dynamic Connectivity}
\label{sec:dynconn}

Now, we apply the equivalence-class data structures developed in Section~\ref{sec:key} to design dynamic connectivity structures for axis-aligned line segments, disks, and arbitrary line segments in $\R^2$. We follow the basic approach from Section~\ref{sec:recap} but describe how to handle general online updates.

\subsection{Axis-Aligned Line Segments}
\label{sec:dynconn-seg0}

For axis-aligned segments, we maintain the following data structures during a phase:

\begin{itemize}
    \item The equivalence-class data structure.
    \item A decremental data structure for explicitly maintaining the connected components of $S$, along with their sizes. This structure has $\tildeO(n)$ preprocessing time and $\tildeO(n)$ total update time over the entire phase. We achieve this by reducing to decremental graph connectivity under edge deletions~\cite{Thorup99} (using the biclique cover technique from Chan~\cite{Chan06}).
    \item For each component $c$, an orthogonal intersection searching structure with $\tildeO(1)$ query and update time (e.g., \cite{EdelsbrunnerM81}).
\end{itemize}

Throughout the phase, we maintain a proxy graph $H$ in the same way as in Afshani and Chan's method, as described in Section~\ref{sec:recap}.  Because 
updates are online, we need to handle not only insertions and deletions in $Q$ but also deletions in $S$.

\subparagraph*{Insertion/deletion in $Q$.} 
Insertions in $Q$ can be done as before, since the equivalence-class structure already supports insertions to $Q$. By Corollary~\ref{cor:seg0-split-time}, the amortized cost of equivalence class maintenance per insertion to $Q$ is $\tildeO(q^3 + n/q)$. 
The number of edge changes in the proxy graph $H$ is $O(q^4)$; we can afford to reconstruct the connectivity structure of $H$ in $O(q^4)$ time.

For deletions in $Q$, we don't even need to update the equivalence-class data structure.

\subparagraph*{Deletion in $S$.} Suppose that a segment $s\in S$ is deleted in the $i$-th update. Let $c$ be the component that contains $s$. The deletion of $s$ may divide $c$ into smaller sub-components $c_1,\ldots,c_z$, listed in order of decreasing size. Let $m_i=|c_2|+\cdots+|c_z|$. Since $|c_2|,\ldots,|c_z|\leq |c|/2$, the sum $\sum_i m_i$ over the entire phase is $\tildeO(n)$.

First, we update the decremental connectivity structure; this takes $\tildeO(n)$ time over the entire phase. We then report the disks of $c_2,\ldots,c_z$, delete them from the orthogonal intersection searching structure for $c$, and then rebuild the structure for each of $c_2,\ldots,c_z$. This takes $\tildeO(m_i)$ time. We make $c_1$ a singleton class, if it isn't already. 
We delete $c_1$ from the equivalence-class structure.\footnote{
All we want is that if two components are in the same class, they intersect the same elements of $Q$; we don't need the converse.  This explains why it is fine to handle some  classes such as $c_1$ separately, and why we can ignore deletions in $Q$ in the equivalence-class structure.
}
We create at most $q$ singleton classes this way, which is negligible.

Next, we compute the classes among $c_2,\ldots,c_z$ with respect to $Q$ by ``replaying'' the insertions in $Q$ from scratch.
This takes $\tildeO(q^4+m_i)$ time. For each class $L'$ found this way, we find the existing class $L$ that is equivalent to $L'$ (this takes $\OO(q^4)$ total time, by viewing each class as a $q$-bit vector and sorting $O(q^3)$ such vectors). If such a class $L$ exists, then we insert the components in $L'$ to the equivalence-class structure.
This takes $\tildeO(m_i)$ amortized time. Thus, the total time for handling deletions is bounded by $\tildeO(q^5 + n)$, which amortizes to $\tildeO(q^4 + n/q)$. This is the
overall amortized update time. Choosing $q=n^{1/5}$, we achieve update time $\tildeO(n^{4/5})$.

\begin{theorem}
    For $n$ axis-aligned line segments in $\R^2$, there is a dynamic connectivity data structure with $O(1)$ query time and $\tildeO(n^{4/5})$ amortized update time.
\end{theorem}

\subsection{Disks}
\label{sec:dynconn-disk}

The dynamic data structure for disks is similar to the case of axis-aligned segments, except we use a different decremental connectivity structure and a different intersection searching structure.

\begin{itemize}
    \item For decremental connectivity, we use the data structure by Kaplan {\em et al.}~\cite{KaplanKKKMRS22}, which achieves $\tildeO(n)$ preprocessing time and $\tildeO(n)$ total update time over the entire phase.
    \item For intersection searching, we use the data structure by Kaplan {\em et al.}~\cite{KaplanMRSS20}, which achieves $\tildeO(1)$ query and update time.
\end{itemize}

\subparagraph*{Update time.} The update algorithm is identical to the case of axis-aligned segments. By Corollary~\ref{cor:disk-split-time}, the amortized cost of equivalence class maintenance per insertion to $Q$ is $O^*(q^{3/5} n^{4/5} + n/q)$. For deletion, we define $m_i$ for $i=1,\ldots,q$ in the same way as we did for axis-aligned segments. Using the same analysis, the total cost over all deletions in $S$ is bounded by
\begin{equation*}
\begin{aligned}
    O^*\left( \sum_{i=1}^q \left( q^{8/5} m_i^{4/5} + m_i \right) + n \right)
    \:=\: O^*\left( q^{9/5} n^{4/5} + n \right),
\end{aligned}
\end{equation*}
which amortizes to $O^*(q^{4/5} n^{4/5} + n/q)$ per deletion. This is the overall amortized update time. Choosing $q=n^{1/9}$, we achieve update time $O^*(n^{8/9})$.

\begin{theorem}
    \label{thm:disk-basic}
    For $n$ disks in $\R^2$, there is a dynamic connectivity data structure with $O(1)$ query time and $O^*(n^{8/9})$ amortized update time.
\end{theorem}

\subsection{Arbitrary Line Segments}
\label{sec:dynconn-segs}
 
For the case of arbitrary line segments the idea is again similar to the case of axis-aligned segments, but because the time bound for Subproblem~\ref{prb:compute-eq-cls} is superlinear and dependent on the parameter $r$ even when $q$ is small, the settings of parameters are more delicate.  There is also a new issue: we don't have a decremental connectivity structure for arbitrary line segments with both near-linear preprocessing time and near-linear total update time (the best static algorithm requires $O(n^{4/3})$ time~\cite{ChanZ22}).

Fortunately, we observe that in the preprocessing step between every two update phases, no more than $q$ new segments are added to the decremental connectivity structure. Thus, the decremental connectivity structure does not necessarily need $O^*(n)$ preprocessing time. Instead, it only needs to support {\em batch insertion} of at most $q$ segments in near-linear time. Such a data structure has been given by Chan, \Patrascu{}, and Roditty~\cite{ChanPR11}.

\begin{lemma}[\cite{ChanPR11}]
    Let $S$ be a set of $n$ line segments in $\R^2$. There is a decremental connectivity data structure with $\tildeO(n)$ total deletion time over any sequence of deletions, and supports batch insertion of any $q$ line segments in $\tildeO(n + q\sqrt{n})$ time.
\end{lemma}

Since we will choose $q\ll \sqrt{n}$, the $\tildeO(q\sqrt{n})$ term is negligible.

\subparagraph*{Update time.} The update algorithm is again identical to the case of axis-aligned segments. By Corollary~\ref{cor:segs-split-time}, the amortized cost of  equivalence-class maintenance per insertion to $Q$ is $\tildeO(n/r + q^3 + r^4 n/q)$. For deletions in $S$, define $m_i$ for $i=1,\ldots,q$ in the same way as we did for disks. Then the total cost over all deletions is bounded by
\begin{equation*}
    \OOO\left( \sum_{i=1}^q \left( \frac{qm_i}{r} + q^4 + r^4 m_j \right) + n \right) \: =\: \OOO\left( \frac{qn}{r} + q^5 + r^4 n \right),
\end{equation*}
which amortizes to $\OOO(n/r + q^4 + r^4 n/q)$ per deletion. This is the dominating term in the overall amortized update time. 
Choosing $r=n^{1/21}$ and $q=n^{5/21}$, the amortized update time minimizes to $\OOO(n^{20/21})$.

\begin{theorem}
    \label{thm:seg-basic}
    For $n$ arbitrary line segments in $\R^2$, there is a dynamic connectivity data structure with $O(1)$ query time and $O^*(n^{20/21})$ amortized update time.
\end{theorem}

We remark that the same approach also works more generally for fixed-degree algebraic curves in 2D, with a larger exponent of the update bound depending on the degree (as the dimension of the lifted space gets larger in the proof of Lemma~\ref{lem:seg:cuttings}). 

\section{Further Improvement for Disks Using Separators}
\label{apd:sep}

In this section, we give a small improvement to our result for disks (from near $n^{8/9}$ to $n^{7/8}$ update time), interestingly by using separators. Specifically, we will use the following variant of Smith and Wormald's geometric separator theorem. (The original version bounds the number of disks intersecting $\partial B$ by $O(\sqrt{n})$ but assumes the input disks are disjoint.  Our variant instead bounds the stabbing number but does not require disjointness.)
For completeness, we include a proof.

\begin{lemma}[Smith and Wormald's separator lemma~\cite{SmithW98}]
    \label{lem:disk-sep}
    Given a set $S$ of $n$ disks in $\R^2$, there is an axis-parallel square $B$, such that the number of disks inside and the number of disks outside are both at most $4n/5$, and the set of all disks intersecting $\partial B$ can be stabbed by $O(\sqrt{n})$ points.  Furthermore, $B$ can be computed in $\OO(n)$ time.
\end{lemma}

\begin{proof}
    Compute smallest square $B_0$ that contains at least $n/5$ of the disk centers;
    this takes $\OO(n)$ time~\cite{Chan99}.  Say $B_0$ has center $(x,y)$ and side length $r$.  For $t\in\{\frac1b,\frac2b,\ldots,\frac{b-1}b\}$,
    let $B_t$ be the square with center $(x,y)$ and side length $(1+t)r$.
    Since $B_t$ can be covered by 4 squares of side length $<r$, we know that $B_t$ 
    contains at most $4n/5$ centers and at least $n/5$ centers.
    
    A disk of radius $\le r/b$ intersects $\partial B_t$ for at most $O(1)$ choices of $t$.
    Thus, we can find $t\in\{\frac1b,\frac2b,\ldots,\frac{b-1}b\}$ in $O(n)$ time such that $\partial B_t$ intersects $O(n/b)$ disks
    of radius $\le r/b$.
    On the other hand, the disks of radius $>r/b$ intersecting $\partial B_t$ can be stabbed by $O(b)$ points.
    Setting $b=\sqrt{n}$, we see that $B_t$ satisfies the requirement.
\end{proof}

We apply the separator lemma to give a new divide-and-conquer algorithm for computing equivalence classes for disks.  This algorithm is static (not supporting insertions to $Q$).

\begin{lemma}
For $n$ disks in $\R^2$, we can solve Subproblem~\ref{prb:compute-eq-cls} in 
$\OO(n + q^6 + q^2\sqrt{n})$ time.
\end{lemma}
\begin{proof}
We first preprocess each component $c$ so that we can quickly determine whether a query disk intersects $c$.
This reduces to nearest neighbor search with additive weights, and so there is a data structure with $\OO(|c|)$ preprocessing
time and $\OO(1)$ query time.

Given a region $R\subseteq \R^2$ and a set $C$ of components completely inside $R$ with total size $n$, we describe a recursive algorithm to
compute the equivalence classes of the components $C$ with respect to $Q$.  The region $R$ is assumed to be
an orthogonal polygon, possibly with holes, with $\OO(1)$ complexity.  (Initially, $R=\R^2$.)

Define the \emph{special points} to be the intersection points between the boundaries of two disks in $Q$, and the inflection points
(i.e., the topmost, bottommost, leftmost, and rightmost points) of the disks in $Q$.

\subparagraph*{Base case.} Suppose that $R$ contains no special points.
Let $\Gamma$ be arcs formed by intersecting the boundaries of disks in $Q$ with the region $R$.
Since $R$ contains no special points, the arcs in $\Gamma$ are disjoint, and
each arc in $\Gamma$ has endpoints lying on two different edges of $\partial R$ (otherwise, $R$ would contain an inflection point).
Two arcs have the same \emph{type} if their endpoints lie on the same pair of edges of $\partial R$.
There are $O(1)$ types.  We can sort the arcs of the same type. (See Figure~\ref{fig:disk-sep}.)
The arcs of the same type intersecting a component $c$ form a contiguous subsequence (i.e., an interval) in the sorted sequence,
which can be determined by binary search in $\OO(|c|)$ time.
Thus, we can compute the $O(q^3)$ equivalence classes of the components with respect to $\Gamma$ in $\OO(n)$ time.

\begin{figure}
    \centering
    \includegraphics[page=1,width=.3\textwidth]{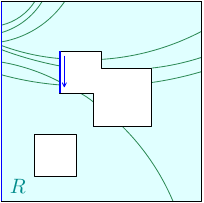}
    \caption{Arcs of each type can be sorted. For example, arcs with endpoints on the two blue edges can be sorted as illustrated by the blue arrow.}
    \label{fig:disk-sep}
\end{figure}

There may be two components $c$ and $c'$ that are equivalent with respect to $\Gamma$ but not equivalent with respect to $Q$.
But this can happen only if $c$ is completely inside a disk $s\in Q$ and $c'$ is completely outside $s$, or vice versa;
in such a case, $c$ and $c'$ cannot intersect any arc in $\Gamma$ (otherwise, the arc would intersect $\partial s$ inside $R$).
For components intersecting no arcs in $\Gamma$, we can put them in the same class if they lie in the same face in the
arrangement of $\Gamma$ inside $R$; there are $O(q)$ such classes, and
this step can be done in $\OO(1)$ time per component via point location~\cite{BergCKO08}.

As a result, we obtain $O(q^3)$ classes, which are a refinement of the actual equivalence classes with respect to $Q$.
We can determine the actual equivalence classes by testing whether each disk in $Q$ intersects (a representative component of)
each class intersects each disk, in $\OO(q\cdot q^3)$ total time.

\subparagraph*{Recursion.}  Suppose that $R$ contains at least one special point.
We apply Smith and Wormald's separator theorem to obtain a square $B$.
We recursively solve the subproblem for $R\cap B$ and the components completely inside $R\cap B$,
and the subproblem for $R\setminus B$ and the components completely inside $R\setminus B$.
There are at most $O(\sqrt{n})$ components intersecting $\partial B$
(because disks stabbed by a common point are in a common component).
We can determine the actual equivalence classes for these ``boundary components''
by  testing whether each disk in $Q$ intersects each boundary component, in $\OO(q\cdot \sqrt{n})$ time.
We can combine the equivalence classes of the two subproblems and the boundary components in $\OO(q\cdot q^3)$ additional time.

\subparagraph*{Analysis.}  
Note that because the recursion has logarithmic depth, each region $R$ generated is indeed  an orthogonal polygon with $\OO(1)$ complexity.
Letting $X$ be the number of special points in $R$, we obtain the
recurrence
\[ T(n,X)\le \max_{n_1+n_2\le n,\ n_1,n_2\le 4n/5, X_1+X_2\le X}  ((T(n_1,X_1)+T(n_2,X_2) + \OO(n+q^4 + q\sqrt{n})),
\]
with the base case $T(n,0)=\OO(n+q^4)$.

The recurrence solves to $T(n,X)=\OO(n + q^4 X + q\sqrt{nX})$.
As the number of special points is $O(q^2)$, we obtain a running time of $\OO(n + q^6 + q^2\sqrt{n})$.
\end{proof}

We can now use the new lemma to speed up our previous
method for dynamic connectivity for disks.
We use the same equivalence-class data structure as before, but make only one change:
during a deletion in $S$, when we compute the equivalence classes among $c_2,\ldots,c_z$, we switch to the above lemma (since this step is a static problem).

\subparagraph*{Update time.} 
The amortized cost per insertion to $Q$ is the same as before, i.e.,  $O^*(q^{3/5} n^{4/5} + n/q)$. With
the new static equivalence-class algorithm,
the total cost over all deletions in $S$ is now bounded by
\begin{equation*}
\begin{aligned}
    O^*\left( \sum_{i=1}^q \left( q^{6} + q^2 \sqrt{m_i} + m_i \right) + n \right)
    \:=\: O^*\left( q^7 + q^{5/2} \sqrt{n} + n \right),
\end{aligned}
\end{equation*}
which amortizes to $O^*(q^6 + q^{3/2}\sqrt{n} + n/q)$ per deletion.
The insertion cost now dominates.  Choosing $q=n^{1/8}$, we achieve update time $O^*(n^{7/8})$.

\begin{theorem}
    For $n$ disks in $\R^2$, there is a dynamic connectivity data structure with $O(1)$ query time and $O^*(n^{7/8})$ amortized update time.
\end{theorem}

\bibliographystyle{plainurl}
\bibliography{reference}

\appendix


\end{document}